\documentclass[letterpaper, 10 pt]{ieeeconf}  
 
\IEEEoverridecommandlockouts 
\usepackage{geometry}
\newgeometry{top=0.75in, bottom=0.75in, right=0.75in, left=0.75in}

\newcommand{\zono}[1]{\langle #1 \rangle}
\usepackage{packagearticleieee}
\usepackage{amssymb}

\newtheorem{theorem}{Theorem}
\newtheorem{proposition}{Proposition}

\usepackage{cite}
\usepackage{booktabs}

 \usepackage{algorithm} 
 \usepackage{algpseudocode}
 \usepackage[utf8]{inputenc}
\usepackage{graphicx}
\usepackage{subcaption}
\usepackage{enumitem}

\makeatletter
\let\NAT@parse\undefined
\makeatother
\usepackage{url}
\usepackage[colorlinks=true,linkcolor=blue!80!black,citecolor=blue!80!black,urlcolor=blue!80!black]{hyperref}
\DeclareMathSymbol{\shortminus}{\mathbin}{AMSa}{"39}

\def\tr{\text{tr}}

\title{\LARGE\bf Transformer-Accelerated Interpolated Data-Driven \\Reachability Analysis from Noisy Data }

\author{Zhen Zhang, Ahmad Hafez, Peng Xie, Yanliang Huang, Wenyuan Wu, and Amr Alanwar
\thanks{All authors are with the TUM School of Computation, Information, and Technology (CIT), Technical University of Munich (TUM), Germany. (Email: $\{$zhenzhang.zhang, a.hafez, p.xie, yanliang.huang,  wenyuan.wu, alanwar$\}$@tum.de)}
}

\begin{document}
\maketitle
\thispagestyle{empty}
\pagestyle{empty}

\begin{abstract}
Data-driven reachability analysis provides guaranteed outer approximations of reachable sets from input-state measurements, yet each propagation step requires a matrix-zonotope multiplication whose cost grows with the horizon length, limiting scalability. We observe that data-driven propagation is inherently step-size sensitive, in the sense that set-valued operators at different discretization resolutions yield non-equivalent reachable sets at the same physical time, a property absent in model-based propagation. Exploiting this multi-resolution structure, we propose Interpolated Reachability Analysis (IRA), which computes a sparse chain of coarse anchor sets sequentially and reconstructs fine-resolution intermediate sets in parallel across coarse intervals. We derive a fully data-driven coarse-noise over-approximation that removes the need for continuous-time system knowledge, prove deterministic outer-approximation guarantees for all interpolated sets, and establish conditional tightness relative to the fine-resolution chain. To replace the remaining matrix-zonotope multiplications in the fine phase, we further develop Transformer-Accelerated IRA (TA-IRA), where an encoder-decoder Transformer is calibrated via split conformal prediction to provide finite-sample pointwise and path-wise coverage certificates. Numerical experiments on a five-dimensional linear system confirm the theoretical guarantees and demonstrate significant computational savings.
\end{abstract}

\section{Introduction}

Reachability analysis computes the set of all states a dynamical system can reach from a given initial set under admissible inputs and bounded disturbances~\cite{althoff2010reachability}. When the system matrices are unknown, data-driven methods construct guaranteed over-approximations directly from measured input-state trajectories~\cite{alanwar2021data,alanwar2023data}. While stochastic formulations yield probabilistic reachability certificates~\cite{thorpe2021sreachtools,griffioen2023data,hashemi2024statistical}, they do not provide worst-case set-containment guarantees. We therefore focus on the deterministic setting, where bounded disturbances are assumed and formal outer approximations are sought for all admissible uncertainty realizations~\cite{park2024data,wang2023data,djeumou2022fly,dietrich2024nonconvex}.

Within this setting, zonotope-based representations are widely used due to their closure under linear maps, Minkowski sums, and Cartesian products~\cite{girard2005reachability}. The method of~\cite{alanwar2023data} encodes all data-consistent system matrices as a matrix zonotope (MZ) and propagates reachable sets via MZ-zonotope multiplication. Richer representations, such as constrained matrix zonotopes~\cite{alanwar2023data} and constrained polynomial matrix zonotopes~\cite{kochdumper2023constrained,zhang2025data}, reduce conservatism but do not eliminate the fundamental computational bottleneck: each multiplication step increases the generator count and necessitates order reduction~\cite{girard2005reachability}, whose cumulative cost grows with the prediction horizon. Classical acceleration strategies, including temporal decomposition~\cite{girard2006efficient,althoff2021set} and GPU parallelization~\cite{bak2017hylaa}, have received limited attention in the data-driven setting because MZ-zonotope propagation is inherently sequential. Transformer architectures~\cite{vaswani2017attention}, which process sequences through a fixed-depth network rather than step-by-step recurrence, offer a promising alternative for approximating iterative set-valued computations, while conformal prediction provides distribution-free finite-sample coverage guarantees for the resulting predictions~\cite{vovk2005algorithmic,shafer2008tutorial}.

A key observation underlying this work is that data-driven reachable-set propagation has a temporal structure that differs fundamentally from its model-based counterpart. In model-based reachability, propagations performed with different discretization step sizes are reconciled through the semigroup property of the underlying dynamics. In contrast, data-driven propagation is carried out via set-valued model operators whose uncertainty is re-instantiated at each step; as a result, different step sizes generally yield distinct reachable sets over the same physical time horizon. Rather than viewing this step-size sensitivity merely as a source of conservatism, we treat it as a source of algorithmic structure.

Building on this observation, we propose Interpolated Reachability Analysis (IRA), a multi-resolution framework for data-driven reachability. IRA first computes reachable sets at a coarse temporal resolution to obtain certified anchor sets, and then reconstructs intermediate fine-resolution reachable sets locally between successive anchors. Because these local interpolation problems are independent across coarse intervals, the fine-resolution phase can be executed fully in parallel. To further reduce computation, we introduce Transformer-Accelerated IRA (TA-IRA), which replaces explicit fine-step interpolation with an encoder-decoder Transformer operating on zonotopic set representations. Since this learned surrogate no longer yields deterministic set containment by construction, we calibrate its predictions using split conformal prediction, thereby obtaining finite-sample pointwise and path-wise coverage guarantees with explicit statistical semantics.

The contributions of this paper are threefold.
First, we formalize the step-size sensitivity of data-driven reachability and show how it gives rise to a new multi-resolution viewpoint that is specific to the data-driven setting.
Second, we introduce IRA, a parallelizable reachability framework that combines coarse anchor propagation with local fine-resolution interpolation, together with a fully data-driven procedure for coarse-step disturbance over-approximation; under the stated assumptions, the resulting interpolated reachable sets are guaranteed to outer-approximate the true reachable sets, and we further characterize a conditional tightness property.
Third, we develop TA-IRA, in which a Transformer with zonotope-aware tokenization autoregressively predicts intermediate reachable sets, and we couple this predictor with conformal calibration to obtain finite-sample pointwise coverage at each substep and path-wise joint coverage over each coarse interval, while further reducing wall-clock time.

The paper is organized as follows. Section~\ref{sec:preliminaries} introduces preliminaries and the problem formulation. Section~\ref{sec:sensitivity} analyzes step-size sensitivity. Section~\ref{sec:ira} presents IRA and its theoretical properties. Section~\ref{sec:ta-ira} develops TA-IRA. Section~\ref{sec:numerical-simulations} provides numerical experiments. Section~\ref{sec:conclusion} concludes.

\section{Preliminaries and Problem Formulation}\label{sec:preliminaries}

\subsection{Notation}
Let $\R$ and $\N$ denote the sets of real and natural numbers, respectively. We write $0_{m \times n}$, $1_{m \times n}$, and $I_n$ for the zero, all-ones, and identity matrices, respectively. For a matrix $A$, $A^\top$ is the transpose, $A^\dagger$ the pseudoinverse, $A_{(i,j)}$ the $(i,j)$-th entry, and $A_{(\cdot,j)}$ the $j$-th column. The operator $\mathrm{blkdiag}(\cdot)$ denotes the block-diagonal matrix formed from its arguments.

\subsection{Set Representations}

\begin{definition}[Zonotope {\cite{kuhn1998rigorously}}] \label{def:zonotopes}
Given a center $c_{\mathcal{Z}} \in \mathbb{R}^{n_x}$ and $\gamma_{\mathcal{Z}} \in \mathbb{N}$ generators collected in
$G_{\mathcal{Z}} = \begin{bmatrix} g_{\mathcal{Z}}^{(1)} & \dots & g_{\mathcal{Z}}^{(\gamma_{\mathcal{Z}})} \end{bmatrix}
\in \mathbb{R}^{n_x \times \gamma_{\mathcal{Z}}}$, a zonotope is defined as
\begin{equation}
\mathcal{Z} = \Bigl\{ x \in \mathbb{R}^{n_x} \;\Big|\;
x \!=\! c_{\mathcal{Z}} + \sum_{i=1}^{\gamma_{\mathcal{Z}}} \alpha^{(i)} g_{\mathcal{Z}}^{(i)},
 \alpha^{(i)} \!\in\![\shortminus 1,\!1] \Bigr\}.
\end{equation}
We write $\mathcal{Z} = \zono{c_{\mathcal{Z}}, G_{\mathcal{Z}}}$.
\end{definition}

Zonotopes are closed under linear maps, Minkowski sums, and Cartesian products~\cite{althoff2010reachability}. Specifically, for $\mathcal{Z}_1=\zono{c_{\mathcal{Z}_1},G_{\mathcal{Z}_1}}$, $\mathcal{Z}_2=\zono{c_{\mathcal{Z}_2},G_{\mathcal{Z}_2}}$, and $L \in \mathbb{R}^{m \times n_x}$, the linear map is given by $L\mathcal{Z}=\zono{Lc_{\mathcal{Z}},LG_{\mathcal{Z}}}$, the Minkowski sum by $\mathcal{Z}_1 \oplus \mathcal{Z}_2=\zono{c_{\mathcal{Z}_1}+c_{\mathcal{Z}_2},[G_{\mathcal{Z}_1},G_{\mathcal{Z}_2}]}$, the Minkowski difference by $\mathcal{Z}_1 \ominus \mathcal{Z}_2 = \{z_1 \in \mathbb{R}^n \mid z_1 \oplus \mathcal{Z}_2 \subseteq \mathcal{Z}_1\}$, and the Cartesian product by $\mathcal{Z}_1 \times \mathcal{Z}_2 = \zono{[c_{\mathcal{Z}_1}^\top, c_{\mathcal{Z}_2}^\top]^\top, \mathrm{blkdiag}(G_{\mathcal{Z}_1}, G_{\mathcal{Z}_2})}$. For simplicity, we use the notation $+$ instead of $\oplus$ to denote the Minkowski sum as the type can be determined from the context. Similarly, we use $\mathcal{Z}_1 - \mathcal{Z}_2$ to denote $\mathcal{Z}_1 + (-1)\mathcal{Z}_2$, not the Minkowski difference.

\begin{definition}[Matrix Zonotope {\cite[p.~52]{althoff2010reachability}}] \label{def:matzonotopes}
Given a center matrix $C_{\mathcal{M}} \in \mathbb{R}^{n_x \times p}$ and $\gamma_{\mathcal{M}} \in \mathbb{N}$ generator matrices $\tilde{G}_{\mathcal{M}}=\begin{bmatrix} G_{\mathcal{M}}^{(1)}&\dots&G_{\mathcal{M}}^{(\gamma_{\mathcal{M}})} \end{bmatrix} \in \mathbb{R}^{n_x \times (p \gamma_{\mathcal{M}})}$, a matrix zonotope is defined as
\begin{equation}
	\mathcal{M} {=} \Big\{\! X \in \mathbb{R}^{n_x \times p} \Big| X {=} C_{\mathcal{M}}\! + \!\sum_{i=1}^{\gamma_{\mathcal{M}}} \alpha^{(i)}  G_{\mathcal{M}}^{(i)}  ,\!
	\alpha^{(i)} \!\in\![\shortminus 1,\!1]\!\Big\} .
\end{equation}
We use the shorthand notation $\mathcal{M} = \zono{C_{\mathcal{M}},\tilde{G}_{\mathcal{M}}}$ for a matrix zonotope. 
\end{definition}

\begin{definition}[Self-concatenation~\cite{alanwar2023data}]\label{def:self-concat}
Let $\mathcal{Z}=\zono{c_{\mathcal{Z}},G_{\mathcal{Z}}}\subset\mathbb{R}^{n}$ be a zonotope with $\gamma_{\mathcal{Z}}$ generators. For $T\in\mathbb{N}$, its self-concatenation is $\mathcal{M}_{\mathcal{Z}}^{(T)} = \{ [z_1 \cdots z_T] \in \mathbb{R}^{n\times T} \mid z_j\in\mathcal{Z}\}$, which admits the MZ representation with center $[c_{\mathcal{Z}} \cdots c_{\mathcal{Z}}]$ and $\gamma_{\mathcal{Z}} T$ generators $G_{\mathcal{M}}^{(i,j)} = [0_{n\times (j-1)},\, g_{\mathcal{Z}}^{(i)},\, 0_{n\times (T-j)}]$.
\end{definition}

\subsection{Assumptions}

We assume access to a single input-state trajectory of length $T+1$. Define the data matrices $X_+ = [x_{(1)}\,\cdots\,x_{(T)}]$, $X_- = [x_{(0)}\,\cdots\,x_{(T-1)}]$, $U_- = [u_{(0)}\,\cdots\,u_{(T-1)}]$\label{eq:offline-data}, $D_- = [X_-^\top\; U_-^\top]^\top$, and $D = [X_+^\top\; X_-^\top\; U_-^\top]^\top$.

\begin{assumption}
\label{ass:zon-noise}
For all $k\in\mathbb{Z}_{\geq 0}$, the disturbance $w_{(k)}$ is bounded by a known zonotope $\mathcal{Z}_w$ containing the origin.
\hfill $\lrcorner$
\end{assumption}


\begin{assumption} \label{ass:rank_D}
For any data matrix $D_-$, we assume that $D_-$ has full row rank, i.e., $\operatorname{rank}(D_-)=n_x+n_u$.
\end{assumption}

By Assumption~\ref{ass:zon-noise} and Definition~\ref{def:self-concat}, the stacked disturbance $W_- = [w_{(0)}\cdots w_{(T-1)}]$ satisfies $W_- \in \mathcal{M}_w = \zono{C_{\mathcal{M}_w},\tilde{G}_{\mathcal{M}_w}}$, where $\mathcal{M}_w$ is the self-concatenation of $\mathcal{Z}_w$ over $T$ columns.

We define the set of all system matrices consistent with the data as $\mathcal{N}_\Sigma \triangleq \bigl\{ [A\;B] \bigm| X_+ = AX_- + BU_- + W_-,\, W_- \in \mathcal{M}_w \bigr\}$\label{eq:Nsig}. By construction, $[A_\tr\;B_\tr]\in\mathcal{N}_{\Sigma}$.

\subsection{Problem Formulation}

We consider the discrete-time linear time-invariant system
\begin{equation}
\label{eq:model-linear}
    x_{(k)} = A_\tr x_{(k-1)} + B_\tr u_{(k-1)} + w_{(k)},
\end{equation}
where $x_{(k)} \in \mathbb{R}^{n_x}$ is the state, $u_{(k)} \in \mathbb{R}^{n_u}$ is the input, and $w_{(k)} \in \mathbb{R}^{n_x}$ is an unknown disturbance. The matrices $A_\tr$ and $B_\tr$ are unknown. Let $\mathcal{X}_0 \subset \mathbb{R}^{n_x}$ be a compact set of initial states, $\mathcal{U}_k \subset \mathbb{R}^{n_u}$ a compact set of admissible inputs, and $\mathcal{Z}_w \subset \mathbb{R}^{n_x}$ a compact disturbance set.

\begin{definition}[Exact Reachable Set]
The exact reachable set at time $N \in \mathbb{N}$ is
\begin{align}
\label{eq:R}
\mathcal{R}_{N} = &\big\{ x_{(N)} \in \mathbb{R}^{n_x} \big|
x_{(k+1)} = A_\tr x_{(k)} + B_\tr u_{(k)} + w_{(k)}, \nonumber \\&
x_{(0)} \in \mathcal{X}_0,
u_{(k)} \in \mathcal{U}_k,
w_{(k)} \in \mathcal{Z}_w,
k = 0,\dots,N-1
\big\}.
\end{align}
\end{definition}

\begin{problem}\label{prob:main}
Given input-state data generated by \eqref{eq:model-linear}, a compact initial set $\mathcal{X}_0$, admissible input sets $\mathcal{U}_k$, and a compact disturbance set $\mathcal{Z}_w$, compute guaranteed outer approximations $\hat{\mathcal{R}}_k$ of the exact reachable sets $\mathcal{R}_k$ over a prescribed finite horizon such that $\mathcal{R}_k \subseteq \hat{\mathcal{R}}_k$ for all $k$, while reducing the computational cost of data-driven reachability analysis through interpolation-based constructions.
\end{problem}

\section{Step-Size Sensitivity of Data-Driven Reachability}\label{sec:sensitivity}

Model-based reachable sets evaluated at a given physical time are invariant to the discretization step size, since the exact system matrices compose consistently across resolutions. In the data-driven setting, this invariance breaks down: the over-approximating model sets depend on the step size, causing the reachable sets at the same physical time to differ across discretizations. We formalize this distinction.

\subsection{Model-Based Step-Size Invariance}

Consider the continuous-time system 
\begin{align} \label{consys}
\dot{x} = A_c x + B_c u + w
\end{align}
discretized at a fine step size $\Delta_f$ and a coarse step size $\Delta_c = N_s \Delta_f$, $N_s \in \N$. The discrete-time matrices are $(A_f, B_f)$ and $(A_c, B_c)$, respectively, with noise~\cite{althoff2010reachability}
\begin{equation}\label{eq:gamma_disc}
\mathcal{Z}_w^{\Delta} = B_w(\Delta) \sigma_w, \quad B_w(\Delta) = A_c^{-1}(e^{A_c \Delta} - I),
\end{equation}
where $\sigma_w$ is the continuous-time noise intensity.

Model-based propagation uses exact matrices satisfying $A_c = A_f^{N_s}$, so the reachable set at any physical time $t = k\Delta_c$ is independent of the step size. The same input zonotope $\mathcal{U}^c = \mathcal{U}^f = \mathcal{U}$ is used at both resolutions. As confirmed in Fig.~\ref{fig:mb_comparison}, the model-based coarse sets with $\mathcal{Z}_w^c = \bigoplus_{i=0}^{N_s-1}A_f^i\mathcal{Z}_w^f$ coincide with fine sets at every shared time; any discrepancy in the data-driven case arises from the model sets $\hat{\mathcal{M}}_\Sigma^c$, $\hat{\mathcal{M}}_\Sigma^f$.

\subsection{Data-Driven Step-Size Dependence}

\begin{lemma}[\cite{alanwar2023data}]
\label{lm:sigmaM}
Consider the offline input-state data \eqref{eq:offline-data} such that Assumption~\ref{ass:rank_D} holds and consider the MZ
\begin{align}
    \mathcal{M}^\Sigma = (X_+ - \mathcal{M}_w) \begin{bmatrix} 
    X_- \\ U_- 
    \end{bmatrix}^{\dagger}.
    \label{eq:zonoAB}
\end{align} 
Then, $\begin{bmatrix} A_\tr & B_\tr \end{bmatrix} \in \mathcal{M}^\Sigma$ and $\mathcal{N}_\Sigma \subseteq \mathcal{M}^\Sigma$. 
\hfill $\lrcorner$
\end{lemma}

The IRA construction employs two data-driven model sets for the fine and coarse discretizations. Let $D_f=(U_{f,-},X_f)$ denote fine-step data, and $D_c=(U_{c,-},X_c)$ the coarse-step data obtained by subsampling $D_f$ every $N_s$ steps, with $D_{f,-} = [X_{f,-}^\top\;U_{f,-}^\top]^\top$ and $D_{c,-} = [X_{c,-}^\top\;U_{c,-}^\top]^\top$.
If both satisfy Assumption~\ref{ass:rank_D}, Lemma~\ref{lm:sigmaM} gives the model sets
\begin{align}
\hat{\mathcal{M}}_\Sigma^f
&=
\left(X_{f,+} - \mathcal{M}_w^f\right)\left(D_{f,-}\right)^\dagger,
\label{eq:Mf}\\
\hat{\mathcal{M}}_\Sigma^c
&=
\left(X_{c,+} - \mathcal{M}_w^c\right)\left(D_{c,-}\right)^\dagger,
\label{eq:Mc}
\end{align}
with $[A_f\; B_f] \in \hat{\mathcal{M}}_\Sigma^f$ and $[A_c\; B_c] \in \hat{\mathcal{M}}_\Sigma^c$. \label{eq:cf-membership}

\begin{remark}
In the IRA setting, Assumption~\ref{ass:rank_D} is required for both $D_{f,-}$ and $D_{c,-}$. Since $D_{c,-}$ is obtained by subsampling the fine-step data, full row rank at the fine resolution does not, in general, imply full row rank at the coarse resolution. In all experiments, this condition was verified for both discretizations prior to constructing $\hat M_\Sigma^f$ and $\hat M_\Sigma^c$.
\end{remark}

In the data-driven setting, each MZ-zonotope multiplication introduces independent scaling factors $\alpha_{(i)} \in [-1,1]$, treated as independent across time steps.

\begin{proposition}\label{prop:dd_depend}
Let $\hat{\mathcal{M}}_\Sigma^c$ and $\hat{\mathcal{M}}_\Sigma^f$ denote the matrix zonotopes constructed from data at step sizes $\Delta_c$ and $\Delta_f$ with $\Delta_c = N_s \Delta_f$, respectively. Let $\hat{\mathcal{R}}_1^c$ denote the data-driven reachable set obtained by a single propagation step of size $\Delta_c$, and let $\hat{\mathcal{R}}_{N_s}^f$ denote the set obtained by $N_s$ successive propagation steps of size $\Delta_f$. Both sets correspond to the same physical time $t = \Delta_c$. Then, in general,
\begin{equation}
\hat{\mathcal{R}}_1^c \neq \hat{\mathcal{R}}_{N_s}^f.
\end{equation}
The discrepancy arises because the number of independent scaling factors differs between the two discretizations: a single coarse step introduces one set of factors, whereas $N_s$ fine steps introduce $N_s$ independent sets.
\end{proposition}

\section{Interpolated Reachability Analysis}\label{sec:ira}


Using the model sets $\hat{\mathcal{M}}_\Sigma^c$ and $\hat{\mathcal{M}}_\Sigma^f$ from \eqref{eq:Mf}--\eqref{eq:Mc}, IRA computes guaranteed outer approximations at every fine time point $t=j\Delta_f$, $j=0,\ldots,KN_s$, while reducing computational cost.

\subsection{Phase 1: Anchor Computation}

At coarse time points $t_k = k \Delta_c$ for $k = 0, 1, \ldots, K$, data-driven reachable sets are computed sequentially using the coarse model set. These sets serve as fixed reference endpoints for the subsequent interpolation and are therefore termed anchor sets:
\begin{equation}\label{eq:anchor}
\hat{\mathcal{R}}_{k+1}^{\mathrm{dd}} = \hat{\mathcal{M}}_c^\Sigma (\hat{\mathcal{R}}_k^{\mathrm{dd}} \times \mathcal{U}^c) \oplus \mathcal{Z}_w^c, \quad \hat{\mathcal{R}}_0^{\mathrm{dd}} = \mathcal{X}_0.
\end{equation}
Phase 1 is sequential but involves only $K$ multiplications instead of $KN_s$, reducing the sequential depth.

\subsection{Coarse Noise Estimation from Data}\label{subsec:noise_est}

The anchor computation~\eqref{eq:anchor} requires $\mathcal{Z}_w^c$. Since $A_c$ is unknown in the data-driven setting, $\mathcal{Z}_w^c$ must be estimated from the fine-step model set.

Unrolling $N_s$ fine-step disturbances over one coarse interval yields
\begin{equation}\label{eq:w_c_true}
\mathcal{Z}_w^c = \bigoplus_{i=0}^{N_s - 1} A_f^i \, \mathcal{Z}_w^f.
\end{equation}
Since $A_f$ is unknown, we replace it by the A-block matrix zonotope
\begin{equation}\label{eq:MfA}
\hat{\mathcal{M}}_f^A = \zono{(C_f)_{(\cdot,\,1:n_x)},\; (\tilde{G}_f)_{(\cdot,\,1:n_x)}},
\end{equation}
extracted from the first $n_x$ columns of $\hat{\mathcal{M}}_f^\Sigma$, so that $A_f \in \hat{\mathcal{M}}_f^A$.

The estimated coarse noise is computed iteratively:
\begin{equation}\label{eq:w_c_dd}
\hat{\mathcal{S}}_0 = \mathcal{Z}_w^f,  \hat{\mathcal{S}}_i = \hat{\mathcal{M}}_f^A \hat{\mathcal{S}}_{i-1} \oplus \mathcal{Z}_w^f,  i = 1, \ldots, N_s - 1.
\end{equation}
The data-driven coarse noise estimate is $\hat{\mathcal{Z}}_w^c = \hat{\mathcal{S}}_{N_s - 1}$.

\begin{proposition}\label{prop:noise_overapprox}
The estimated coarse noise $\hat{\mathcal{Z}}_w^c = \hat{\mathcal{S}}_{N_s-1}$ computed via~\eqref{eq:w_c_dd} satisfies $\mathcal{Z}_w^c \subseteq \hat{\mathcal{Z}}_w^c$.
\end{proposition}

\begin{proof}
We proceed by induction. At $i=0$, $\hat{\mathcal{S}}_0 = \mathcal{Z}_w^f \supseteq A_f^0\mathcal{Z}_w^f$. Suppose $\hat{\mathcal{S}}_{i-1} \supseteq \bigoplus_{l=0}^{i-1}A_f^l\mathcal{Z}_w^f$. Since $A_f\in\hat{\mathcal{M}}_f^A$, we have $A_f\hat{\mathcal{S}}_{i-1} \subseteq \hat{\mathcal{M}}_f^A\cdot\hat{\mathcal{S}}_{i-1}$, so $\bigoplus_{l=0}^{i}A_f^l\mathcal{Z}_w^f \subseteq \hat{\mathcal{S}}_i$.
\end{proof}

The recursion~\eqref{eq:w_c_dd} requires $N_s - 1$ multiplications with $\hat{\mathcal{M}}_f^A \in \R^{n_x \times n_x}$, making the cost modest relative to anchor computation.

\subsection{Phase 2: Parallel Interpolation from Anchors}

The intermediate sets within each coarse interval $[t_k, t_{k+1}]$ are then computed independently via
\begin{equation}\label{eq:interp}
\hat{\mathcal{R}}_{k,j}^{\mathrm{int}} = \hat{\mathcal{M}}_f^\Sigma (\hat{\mathcal{R}}_{k,j-1}^{\mathrm{int}} \times \mathcal{U}^f) \oplus \mathcal{Z}_w^f, \quad \hat{\mathcal{R}}_{k,0}^{\mathrm{int}} = \hat{\mathcal{R}}_k^{\mathrm{dd}},
\end{equation}
for $j = 1, \ldots, N_s - 1$. Within each coarse interval $[t_k, t_{k+1}]$, the recursion is initialized from the anchor $\hat{\mathcal{R}}_k^{\mathrm{dd}}$ and uses only the shared fine model set $\hat{\mathcal{M}}_f^\Sigma$. Because different intervals share no intermediate state, all $K$ intervals can be computed in parallel. Moreover, restarting from a fresh anchor at each coarse boundary prevents over-approximation error from accumulating across intervals.

For reference, the fine-grained data-driven chain and the model-based chain are
\begin{align}
\hat{\mathcal{R}}_{j+1}^{\mathrm{fine}} &= \hat{\mathcal{M}}_f^\Sigma (\hat{\mathcal{R}}_j^{\mathrm{fine}} \times \mathcal{U}^f) \oplus \mathcal{Z}_w^f, \quad \hat{\mathcal{R}}_0^{\mathrm{fine}} = \mathcal{X}_0, \label{eq:fine_chain} \\
\mathcal{R}_{j+1}^{\mathrm{mb}} &= A_f \, \mathcal{R}_j^{\mathrm{mb}} \oplus B_f \, \mathcal{U}^f \oplus \mathcal{Z}_w^f, \quad \mathcal{R}_0^{\mathrm{mb}} = \mathcal{X}_0. \label{eq:model_reach}
\end{align}
The fine chain~\eqref{eq:fine_chain} is sequential over $KN_s$ steps; IRA reduces the sequential depth to $K$ coarse steps plus a parallel phase of $N_s - 1$ fine steps. The complete procedure, including the TA-IRA variant, is given in Algorithm~\ref{alg:unified}.

\subsection{Over-Approximation Guarantee}

\begin{theorem}[Over-Approximation Guarantee]\label{thm:main}
Let Assumptions~\ref{ass:zon-noise} and~\ref{ass:rank_D} hold. Let $\mathcal{R}^*(t)$ denote the true reachable set at time $t$, i.e., the set of all states reachable from $\mathcal{X}_0$ under all admissible inputs $u \in \mathcal{U}$ and noise $w \in \mathcal{Z}_w$. Then for every coarse interval $k = 0, 1, \ldots, K-1$ and substep $j = 1, \ldots, N_s - 1$, the interpolated set satisfies
\begin{equation}
\mathcal{R}^*(t_{k,j}) \subseteq \hat{\mathcal{R}}_{k,j}^{\mathrm{int}}.
\end{equation}
\end{theorem}

\begin{proof}
We proceed by induction on the global fine-step index within each coarse interval.


\emph{Interpolation validity.} Fix a coarse interval $k$. We have $\hat{\mathcal{R}}_{k,0}^{\mathrm{int}} = \hat{\mathcal{R}}_k^{\mathrm{dd}} \supseteq \mathcal{R}^*(t_k)$. Suppose $\hat{\mathcal{R}}_{k,j-1}^{\mathrm{int}} \supseteq \mathcal{R}^*(t_{k,j-1})$ for some $j \geq 1$. Let $x \in \mathcal{R}^*(t_{k,j})$. By definition, there exist $\hat{x} \in \mathcal{R}^*(t_{k,j-1})$, $u \in \mathcal{U}^f$, and $w \in \mathcal{Z}_w^f$ such that $x = A_f \hat{x} + B_f u + w$.

Since $\begin{bmatrix}A_f & B_f\end{bmatrix} \in \hat{\mathcal{M}}_f^\Sigma$ and $\hat{x} \in \hat{\mathcal{R}}_{k,j-1}^{\mathrm{int}}$ by the inductive hypothesis,
\begin{equation}
x = A_f \hat{x} + B_f u + w \in \hat{\mathcal{M}}_f^\Sigma (\hat{\mathcal{R}}_{k,j-1}^{\mathrm{int}} \times \mathcal{U}^f) \oplus \mathcal{Z}_w^f = \hat{\mathcal{R}}_{k,j}^{\mathrm{int}},
\end{equation}
completing the induction.
\end{proof}

\subsection{Tightness Analysis}

To characterize when IRA interpolated sets are no larger than the fine DD chain, we first establish that the propagation operator preserves set inclusion.

\begin{proposition}[Monotonicity of Set Propagation]\label{prop:mono}
Let $\mathcal{Z}_1 \subseteq \mathcal{Z}_2$ be two zonotopes and let $\mathcal{N}$ be a MZ. Then
\begin{equation}
\mathcal{N} (\mathcal{Z}_1 \times \mathcal{U}) \oplus \mathcal{Z}_w \subseteq \mathcal{N} (\mathcal{Z}_2 \times \mathcal{U}) \oplus \mathcal{Z}_w.
\end{equation}
\end{proposition}

\begin{proof}
Since $\mathcal{N} \mathcal{Z} = \{Mz \mid M \in \mathcal{N}, z \in \mathcal{Z}\}$, if $\mathcal{Z}_1 \subseteq \mathcal{Z}_2$ then $\{Mz \mid z \in \mathcal{Z}_1{\times}\mathcal{U}\} \subseteq \{Mz \mid z \in \mathcal{Z}_2{\times}\mathcal{U}\}$, and the Minkowski sum with $\mathcal{Z}_w$ preserves containment.
\end{proof}

The following result shows that if the coarse anchor $\hat{\mathcal{R}}_k^{\mathrm{dd}}$ at time $t_k$ is contained in the fine DD chain set $\hat{\mathcal{R}}_{kN_s}^{\mathrm{fine}}$ at the same physical time, then the IRA interpolated sets $\hat{\mathcal{R}}_{k,j}^{\mathrm{int}}$ remain contained in the corresponding fine DD sets $\hat{\mathcal{R}}_{kN_s+j}^{\mathrm{fine}}$ for all substeps within that interval.

\begin{proposition}[Conditional Tightness]\label{prop:tight}
If $\hat{\mathcal{R}}_k^{\mathrm{dd}} \subseteq \hat{\mathcal{R}}_{kN_s}^{\mathrm{fine}}$ at coarse time $t_k$, then
\begin{equation}\label{eq:tight}
\hat{\mathcal{R}}_{k,j}^{\mathrm{int}} \subseteq \hat{\mathcal{R}}_{kN_s + j}^{\mathrm{fine}}, \quad j = 1, \ldots, N_s - 1.
\end{equation}
\end{proposition}

\begin{proof}
At $j = 0$, $\hat{\mathcal{R}}_{k,0}^{\mathrm{int}} = \hat{\mathcal{R}}_k^{\mathrm{dd}} \subseteq \hat{\mathcal{R}}_{kN_s}^{\mathrm{fine}}$ by hypothesis. The recursions \eqref{eq:interp} and \eqref{eq:fine_chain} apply the same operator $\mathcal{Z} \mapsto \hat{\mathcal{M}}_f^\Sigma (\mathcal{Z} \times \mathcal{U}^f) \oplus \mathcal{Z}_w^f$. By Proposition~\ref{prop:mono}, containment at step $j-1$ implies containment at step $j$, yielding \eqref{eq:tight} by induction.
\end{proof}

\begin{remark}\label{rem:tightness_premise}
The premise $\hat{\mathcal{R}}_k^{\mathrm{dd}} \subseteq \hat{\mathcal{R}}_{kN_s}^{\mathrm{fine}}$ is governed by two competing effects. On the one hand, the fine chain performs $kN_s$ MZ-zonotope multiplications, each of which introduces independent generator scaling factors and requires order reduction; the resulting over-approximation error accumulates over all $kN_s$ steps (the wrapping effect). The coarse chain performs only $k$ such multiplications, producing significantly less accumulated error. On the other hand, the data-driven coarse noise estimate $\hat{\mathcal{Z}}_w^c \supseteq \mathcal{Z}_w^c$ (Proposition~\ref{prop:noise_overapprox}) inflates the coarse anchors beyond what the true coarse noise would yield. For small $k$, this noise inflation may dominate, causing the coarse anchor to exceed the fine chain. As $k$ grows, the wrapping effect in the fine chain grows as $\mathcal{O}(kN_s)$ while the coarse noise inflation remains constant per step, so the premise becomes increasingly likely to hold. Moreover, the condition is checkable at runtime: after computing each anchor $\hat{\mathcal{R}}_k^{\mathrm{dd}}$, one can verify interval hull containment in $\hat{\mathcal{R}}_{kN_s}^{\mathrm{fine}}$ in $\mathcal{O}(n_x)$ time. The experimental results in Table~\ref{tab:tradeoff} confirm that for moderate $K$ the width ratio falls below $1.0$, indicating that the premise holds in practice.
\hfill $\lrcorner$
\end{remark}

\subsection{Computational Complexity}

Let $C_\cdot(\gamma, h)$ denote the per-step MZ-zonotope multiplication cost. The full fine chain~\eqref{eq:fine_chain} is strictly sequential over $K N_s$ steps, with total work and wall-clock depth
\begin{equation}
\mathcal{C}_{\mathrm{fine}} = K N_s \cdot C_\cdot(\gamma_f, h_f).
\end{equation}

The IRA approach decomposes this into two phases with different parallelism properties. The total work is
\begin{equation}\label{eq:cost_ira}
\mathcal{C}_{\mathrm{IRA}} = \underbrace{K \, C_\cdot(\gamma_c, h_c)}_{\text{Phase 1 (sequential)}} + \underbrace{K(N_s - 1) C_\cdot(\gamma_f, h_f)}_{\text{Phase 2 (parallelizable)}},
\end{equation}
while the parallel wall-clock depth on $K$ processors is
\begin{equation}\label{eq:depth_ira}
\mathcal{D}_{\mathrm{IRA}} = \underbrace{K \, C_\cdot(\gamma_c, h_c)}_{\text{sequential anchors}} + \underbrace{(N_s - 1) C_\cdot(\gamma_f, h_f)}_{\text{longest parallel chain}}.
\end{equation}
The parallel speedup over the fine chain is
\begin{equation}
\frac{\mathcal{C}_{\mathrm{fine}}}{\mathcal{D}_{\mathrm{IRA}}} \approx \frac{K N_s}{K + N_s - 1},
\end{equation}
assuming $C_\cdot(\gamma_c, h_c) \approx C_\cdot(\gamma_f, h_f)$, approaching $N_s$ as $K$ grows.

\section{Transformer-Accelerated Interpolation}\label{sec:ta-ira}

\subsection{Zonotope Token Representation}

Each zonotope $\mathcal{Z} = \langle c, [g_1, \ldots, g_h] \rangle \subset \R^{n_x}$ is first reduced to $\kappa = \rho n_x$ generators~\cite{girard2005reachability} and then tokenized as
\begin{equation}\label{eq:tokenize}
\mathrm{tok}(\mathcal{Z}) = [c,\; g_1,\; \ldots,\; g_\kappa] \in \R^{(\kappa+1) \times n_x}.
\end{equation}
Each column is a token of dimension $n_x$, augmented with a time fraction $\tau = t / T$, yielding $d_{\mathrm{tok}} = n_x + 1$ and $L_z = \kappa + 1$ tokens per zonotope. The inverse mapping $\mathrm{tok}^{-1}$ recovers the predicted zonotope.

\subsection{Encoder-Decoder Transformer Architecture}

An encoder-decoder Transformer processes the tokenized anchor pair. The encoder input is
\begin{equation}\label{eq:enc_input}
\mathbf{S}_{\mathrm{enc}} = \big[\mathrm{tok}(\hat{\mathcal{R}}_k^{\mathrm{dd}}),\; \mathrm{tok}(\hat{\mathcal{R}}_{k+1}^{\mathrm{dd}}) \big] \in \R^{2L_z \times d_{\mathrm{tok}}},
\end{equation}
with bidirectional self-attention across all $2L_z$ tokens.

Each token receives additive type, position, and step embeddings. The decoder uses $L_z$ learnable queries with causal self-attention and cross-attention to the encoder, producing
\begin{equation}\label{eq:dec_output}
\hat{\mathbf{S}}_{\mathrm{dec}} = f_\theta(\mathbf{S}_{\mathrm{enc}}) \in \R^{L_z \times d_{\mathrm{tok}}},
\end{equation}
and the predicted zonotope is $\hat{\mathcal{R}}_{k,j}^{\mathrm{TF}} = \mathrm{tok}^{-1}(\hat{\mathbf{S}}_{\mathrm{dec}})$.

\subsection{Autoregressive Inference}

Intermediate sets are predicted autoregressively:
\begin{equation}\label{eq:autoregressive}
\hat{\mathcal{R}}_{k,j}^{\mathrm{TF}} = \mathrm{tok}^{-1}\big(f_\theta([\mathrm{tok}(\hat{\mathcal{R}}_{k,j-1}^{\mathrm{TF}}),\; \mathrm{tok}(\hat{\mathcal{R}}_{k+1}^{\mathrm{dd}})])\big),
\end{equation}
with $\hat{\mathcal{R}}_{k,0}^{\mathrm{TF}} = \hat{\mathcal{R}}_k^{\mathrm{dd}}$. Each prediction conditions on its predecessor while the endpoint anchor $\hat{\mathcal{R}}_{k+1}^{\mathrm{dd}}$ remains fixed.

\subsection{Training Data Generation}

Training data is generated by computing fine DD chains from geometrically augmented initial sets $\mathcal{X}_0^{(i)}$. Each chain of $KN_s$ steps yields $K(N_s - 1)$ autoregressive training pairs, where the encoder input is the tokenized pair of current set and interval endpoint, and the decoder target is the next fine-step set. The training loss is
\begin{equation}\label{eq:loss}
\mathcal{L}(\theta) = \frac{1}{N_{\mathrm{train}}} \sum_{i=1}^{N_{\mathrm{train}}} \| f_\theta(\mathbf{S}_{\mathrm{enc}}^{(i)}) - \mathbf{S}_{\mathrm{dec}}^{(i)} \|_F^2.
\end{equation}
For conformal calibration, $N_{\mathrm{traj}}$ trajectories are simulated from each augmented initial set.

\subsection{Conformal Prediction Calibration}


The Transformer prediction $\hat{\mathcal{R}}_{k,j}^{\mathrm{TF}}$ is a point estimate that does not inherently guarantee containment of the true reachable set. To provide a finite-sample coverage certificate, we employ split conformal prediction~\cite{hashemi2023data} with a max-norm nonconformity score. It is important to note that conformal prediction provides pointwise statistical coverage under exchangeability, which is weaker than deterministic set containment. The latter is not implied by Theorem~\ref{thm:conformal}; when deterministic set containment is required, one may revert to exact IRA, whose guarantee is given by Theorem~\ref{thm:main}.

The calibration is performed on a held-out subset of fine-step training chains, using the same fine-step prompts on which the network is trained. This ensures that the calibration and training distributions are identical, so that the conformal quantile accurately reflects the model prediction quality without conflating it with the coarse-to-fine distribution shift. For each calibration instance $i$, let $\{x_{k,j}^{(i,m)}\}_{m=1}^{N_{\mathrm{traj}}}$ denote the simulated trajectory states at the fine time step corresponding to substep $j$. The predicted zonotope $\hat{\mathcal{R}}_{k,j}^{\mathrm{TF}}$ is evaluated against these trajectories by computing the max-norm nonconformity score
\begin{equation}\label{eq:score}
s_i = \max_{d=1,\ldots,n_x}\; \max_{m=1,\ldots,N_{\mathrm{traj}}} \max\big( \hat{\ell}_{j,d}^{(i)} - x_{j,d}^{(i,m)},\; x_{j,d}^{(i,m)} - \hat{u}_{j,d}^{(i)} \big),
\end{equation}
where $[\hat{\ell}^{(i)}, \hat{u}^{(i)}]$ is the interval hull of $\hat{\mathcal{R}}_{k,j}^{\mathrm{TF}}$. The score $s_i$ measures the worst-case violation across all dimensions and trajectory samples, ensuring that a single scalar quantile suffices for joint coverage.

For a desired coverage level $1 - \delta$, the conformal quantile is
\begin{equation}\label{eq:quantile}
\hat{q} = \max\bigg(0,\;\mathrm{Quantile}\Big(\{s_i\}_{i=1}^{N_{\mathrm{cal}}},\; \frac{\lceil (N_{\mathrm{cal}}+1)(1-\delta) \rceil}{N_{\mathrm{cal}}} \Big)\bigg).
\end{equation}

The conformalized prediction set is obtained by inflating the predicted zonotope uniformly along each coordinate axis:
\begin{equation}\label{eq:conformal_set}
\tilde{\mathcal{R}}_{k,j}^{\mathrm{TF}} = \hat{\mathcal{R}}_{k,j}^{\mathrm{TF}} \oplus \langle 0, \hat{q} \cdot I_{n_x} \rangle.
\end{equation}

\begin{theorem}[Pointwise-in-Time Coverage]\label{thm:conformal}
Let $\tilde{\mathcal{R}}_{k,j}^{\mathrm{TF}}$ be the conformalized Transformer prediction~\eqref{eq:conformal_set} with coverage level $1-\delta$. Under exchangeability of the calibration and test data, for a new state sample $x_{k,j}$ drawn from the same system dynamics at substep $(k,j)$,
\begin{equation}\label{eq:coverage}
\Pr\big[x_{k,j} \in \tilde{\mathcal{R}}_{k,j}^{\mathrm{TF}} \big] \geq 1 - \delta.
\end{equation}
\end{theorem}

\begin{proof}
By the standard guarantee of split conformal prediction~\cite{vovk2005algorithmic}, for any new test instance drawn exchangeably with the calibration data, the probability that its nonconformity score exceeds $\hat{q}$ is at most $\delta$. Since $s_i$ is defined as the maximum violation across all dimensions, $s \leq \hat{q}$ implies that every dimension $d$ satisfies $x_d \in [\hat{\ell}_d - \hat{q},\, \hat{u}_d + \hat{q}]$, which is precisely the interval hull of $\tilde{\mathcal{R}}_{k,j}^{\mathrm{TF}}$. Hence $x \in \tilde{\mathcal{R}}_{k,j}^{\mathrm{TF}}$.
\end{proof}


This pointwise coverage guarantee bounds the failure probability at any single substep. To obtain joint coverage over an entire coarse interval, we extend the score to a path-wise variant.

\begin{corollary}[Path-Wise Coverage]\label{cor:pathwise}
Define the path-wise nonconformity score
\begin{equation}\label{eq:pathscore}
\begin{aligned}
s_i^{\mathrm{path}}
= \max_{\substack{j=1,\ldots,N_s-1\\ d=1,\ldots,n_x\\ m=1,\ldots,N_{\mathrm{traj}}}}
\max \big(& \hat{\ell}_{j,d}^{(i)} - x_{j,d}^{(i,m)},  x_{j,d}^{(i,m)} - \hat{u}_{j,d}^{(i)} \big).
\end{aligned}
\end{equation}
and let $\hat{q}^{\mathrm{path}}$ be the conformal quantile computed from $\{s_i^{\mathrm{path}}\}_{i=1}^{N_{\mathrm{cal}}}$ at level $1-\delta$. Then for a new test instance, the inflated sets $\tilde{\mathcal{R}}_{k,j}^{\mathrm{TF}} = \hat{\mathcal{R}}_{k,j}^{\mathrm{TF}} \oplus \langle 0, \hat{q}^{\mathrm{path}} I_{n_x} \rangle$ satisfy
\begin{equation}
\Pr\Big[\forall j \in \{1,\ldots,N_s{-}1\}:\; x_{k,j} \in \tilde{\mathcal{R}}_{k,j}^{\mathrm{TF}} \Big] \geq 1 - \delta.
\end{equation}
\end{corollary}

\begin{proof}
The score $s_i^{\mathrm{path}}$ aggregates the maximum violation across all substeps and dimensions within a single coarse interval. By the standard conformal guarantee, $\Pr[s^{\mathrm{path}} \leq \hat{q}^{\mathrm{path}}] \geq 1 - \delta$. The event $s^{\mathrm{path}} \leq \hat{q}^{\mathrm{path}}$ implies that every substep $j$ and every dimension $d$ satisfies $x_{j,d} \in [\hat{\ell}_{j,d} - \hat{q}^{\mathrm{path}},\, \hat{u}_{j,d} + \hat{q}^{\mathrm{path}}]$, so the entire sub-trajectory lies within the inflated sets.
\end{proof}

The pointwise and path-wise guarantees correspond to two operating modes of TA-IRA: the pointwise variant (Theorem~\ref{thm:conformal}) provides per-step coverage with tighter inflation, while the path-wise variant (Corollary~\ref{cor:pathwise}) guarantees joint coverage over all substeps within a coarse interval at the cost of a slightly larger quantile. Both are distinct from the deterministic set-containment guarantee of Theorem~\ref{thm:main}; when deterministic containment is required, one may revert to exact IRA.

The validity of the conformal transfer from fine-step calibration to coarse-anchor inference is verified empirically in Section~\ref{sec:numerical-simulations}.

\begin{algorithm}[t]
\caption{Interpolated Reachability Analysis (IRA / TA-IRA)}
\label{alg:unified}
\textbf{Input:} Initial set $\mathcal{X}_0$, input set $\mathcal{U}$, coarse noise $\mathcal{Z}_w^c$, fine noise $\mathcal{Z}_w^f$, coarse model set $\hat{\mathcal{M}}_c^\Sigma$, fine model set $\hat{\mathcal{M}}_f^\Sigma$, number of coarse steps $K$, substeps per interval $N_s$; \emph{optional:} trained predictor $f_\theta$, conformal quantile $\hat{q}$ \\
\textbf{Output:} Over-approximated reachable sets at all fine time points

\noindent\textbf{Phase~1: Compute coarse anchors (sequential)}
\begin{algorithmic}[1]
\State $\hat{\mathcal{R}}_0^{\mathrm{dd}} \gets \mathcal{X}_0$
\For{$k = 0$ to $K-1$}
  \State $\hat{\mathcal{R}}_{k+1}^{\mathrm{dd}} \gets \hat{\mathcal{M}}_c^\Sigma (\hat{\mathcal{R}}_k^{\mathrm{dd}} \times \mathcal{U}^c) \oplus \mathcal{Z}_w^c$
\EndFor
\end{algorithmic}
\noindent\textbf{Phase~2: Interpolation from anchors}
\begin{algorithmic}[1]
\For{$k = 0$ to $K-1$} \textbf{in parallel}
  \State $\hat{\mathcal{R}}_{k,0} \gets \hat{\mathcal{R}}_k^{\mathrm{dd}}$
  \For{$j = 1$ to $N_s - 1$}
    \If{$f_\theta$ available}
      \State $\hat{\mathcal{R}}_{k,j}^{\mathrm{TF}} \gets \mathrm{tok}^{-1}(f_\theta([\mathrm{tok}(\hat{\mathcal{R}}_{k,j-1}),\; \mathrm{tok}(\hat{\mathcal{R}}_{k+1}^{\mathrm{dd}})]))$
      \State $\hat{\mathcal{R}}_{k,j} \gets \hat{\mathcal{R}}_{k,j}^{\mathrm{TF}} \oplus \langle 0, \hat{q} \cdot I_{n_x} \rangle$
    \Else
      \State $\hat{\mathcal{R}}_{k,j} \gets \hat{\mathcal{M}}_f^\Sigma (\hat{\mathcal{R}}_{k,j-1} \times \mathcal{U}^f) \oplus \mathcal{Z}_w^f$
    \EndIf
  \EndFor
\EndFor
\end{algorithmic}
\textbf{Return} $\{\hat{\mathcal{R}}_k^{\mathrm{dd}}\}_{k=0}^K,\; \{\hat{\mathcal{R}}_{k,j}\}_{k,j}$
\end{algorithm}

Algorithm~\ref{alg:unified} supports two modes: IRA provides deterministic set containment (Theorem~\ref{thm:main}); TA-IRA replaces Phase~2 with Transformer inference and yields pointwise coverage (Theorem~\ref{thm:conformal}) at lower cost. Without a trained model, the algorithm falls back to IRA.

\subsection{Computational Advantage}\label{sec:comp_advantage}

The exact Phase~2 requires $K(N_s - 1)$ MZ-zonotope multiplications, each with cost growing quadratically in generator count. TA-IRA replaces each multiplication with a single forward pass through $f_\theta$ at cost $\mathcal{O}(L(2L_z)^2 d_{\mathrm{model}})$, which is constant since $L_z = \kappa + 1$ is fixed by tokenization.

\section{Numerical simulations}\label{sec:numerical-simulations}

\begin{figure*}[t]
    \centering
    \begin{subfigure}[t]{0.32\textwidth}
        \centering
        \includegraphics[width=\textwidth]{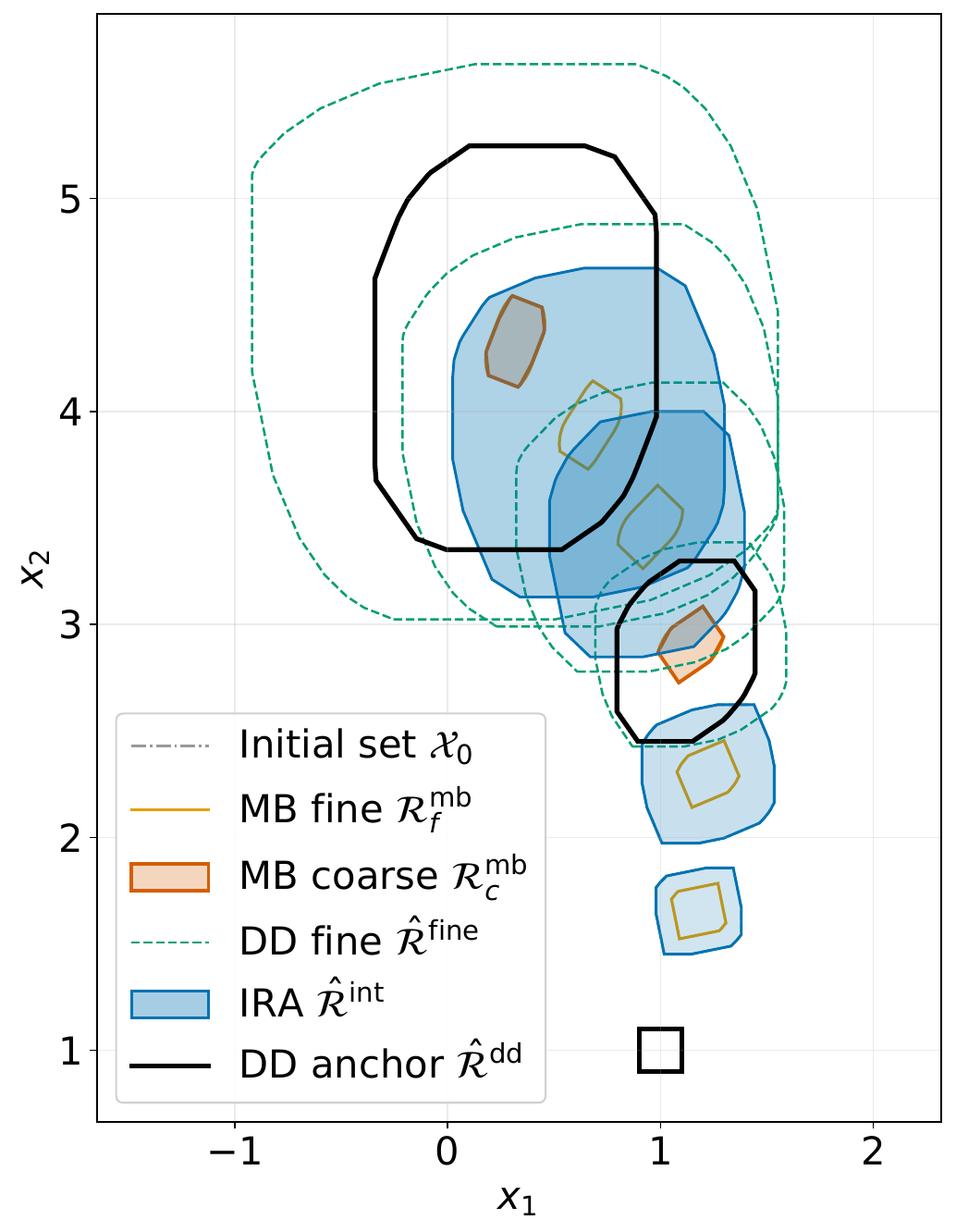}
        \caption{}
        \label{fig:ira_x1x2}
    \end{subfigure}
    \hfill
    \begin{subfigure}[t]{0.32\textwidth}
        \centering
        \includegraphics[width=\textwidth]{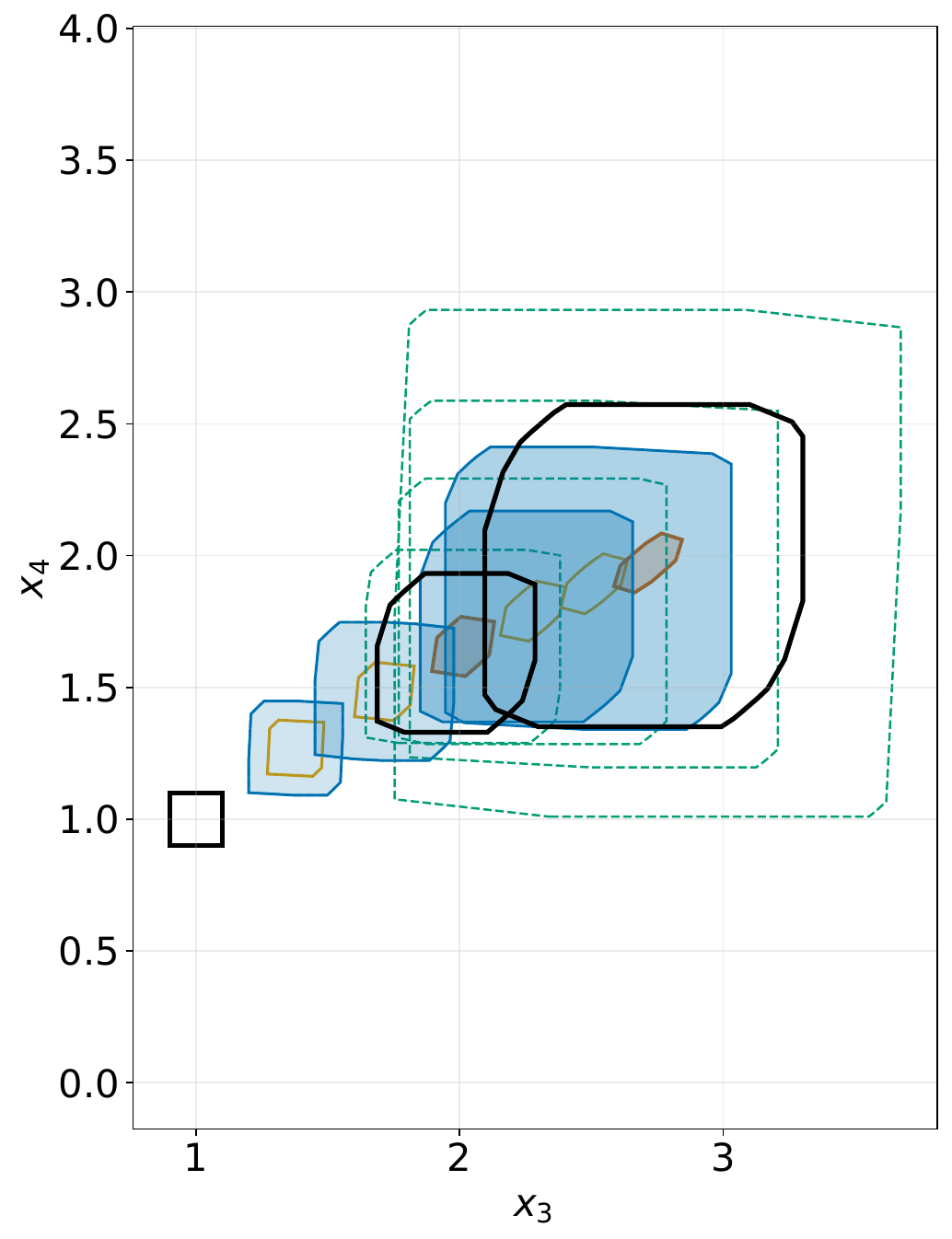}
        \caption{}
        \label{fig:ira_x3x4}
    \end{subfigure}
    \hfill
    \begin{subfigure}[t]{0.32\textwidth}
        \centering
        \includegraphics[width=\textwidth]{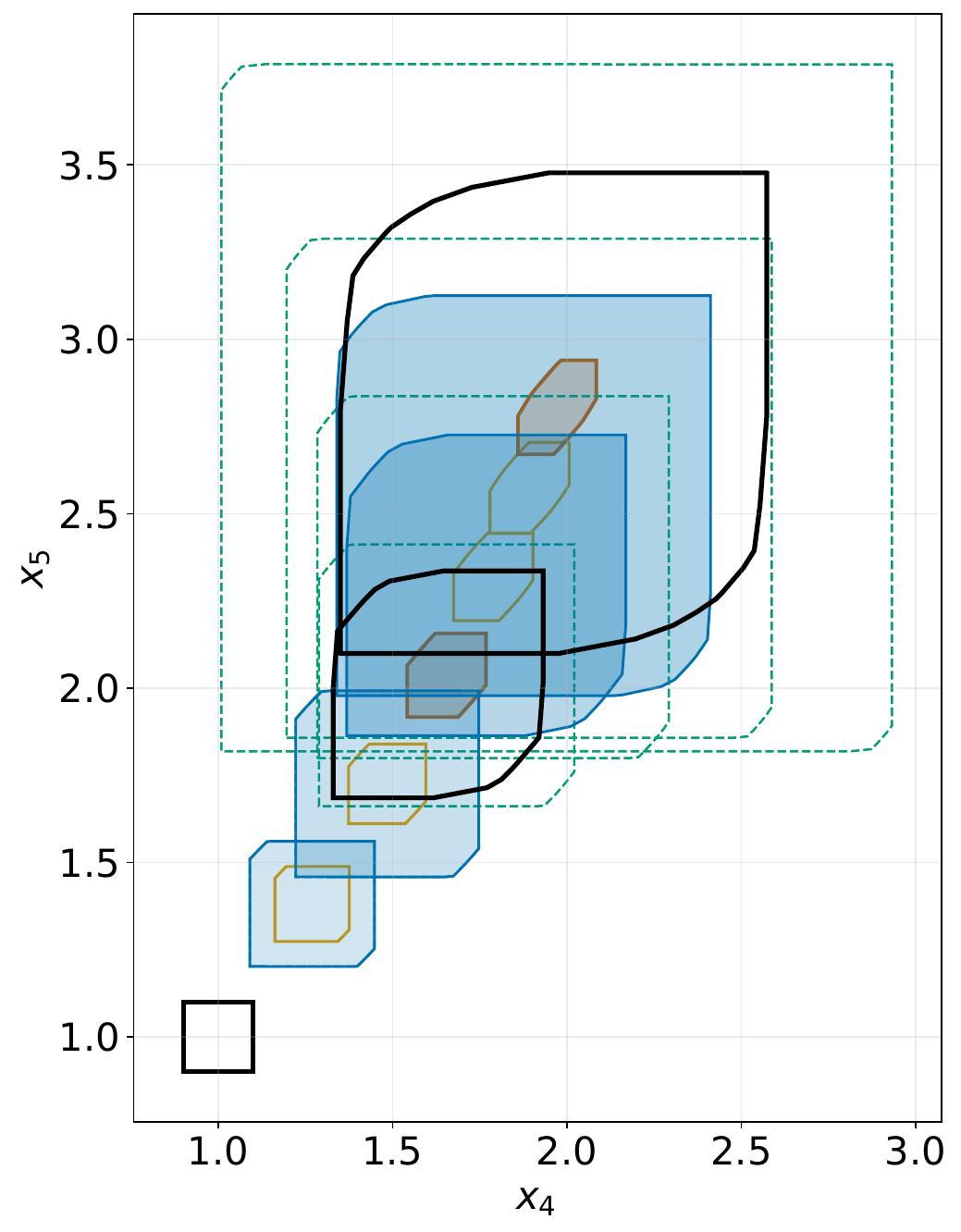}
        \caption{}
        \label{fig:ira_x4x5}
    \end{subfigure}
    \caption{IRA interpolation with $K=2$, $N_s=3$ for three state-space projections.}
    \label{fig:mb_comparison}
\end{figure*}

\begin{figure*}[t]
    \centering
    \begin{subfigure}[t]{0.32\textwidth}
        \centering
        \includegraphics[width=\textwidth]{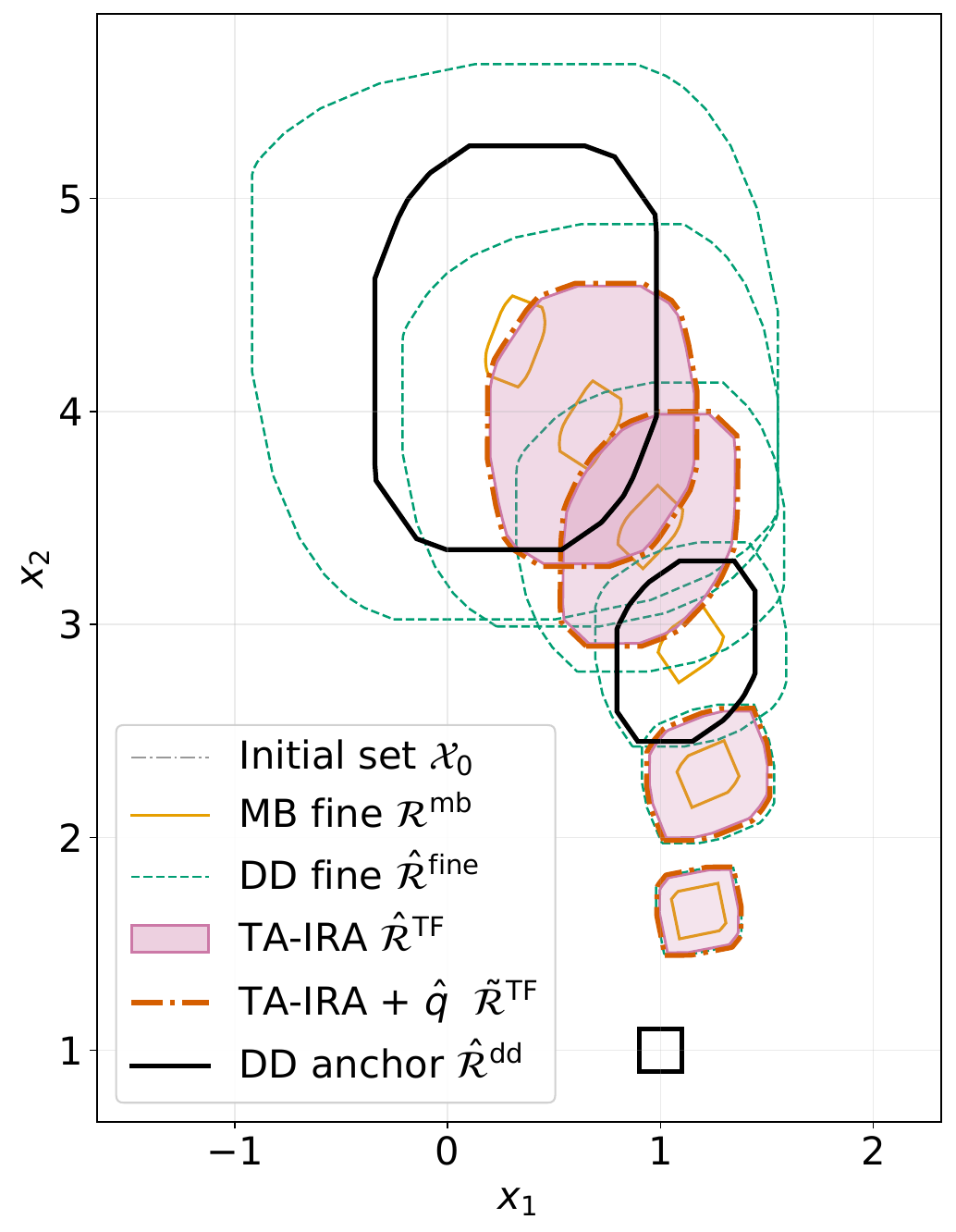}
        \caption{}
        \label{fig:ta_ira_x1x2}
    \end{subfigure}
    \hfill
    \begin{subfigure}[t]{0.32\textwidth}
        \centering
        \includegraphics[width=\textwidth]{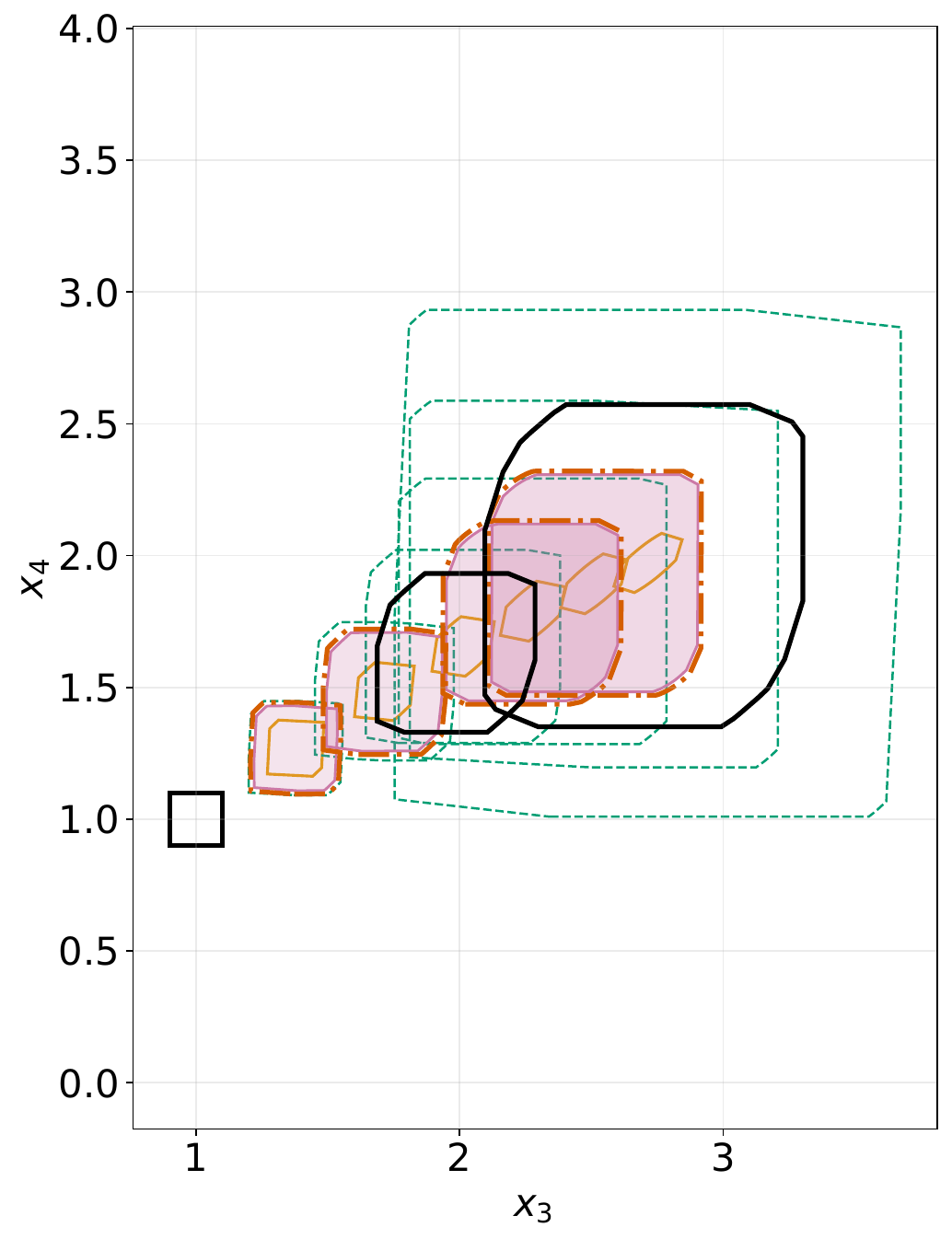}
        \caption{}
        \label{fig:ta_ira_x3x4}
    \end{subfigure}
    \hfill
    \begin{subfigure}[t]{0.32\textwidth}
        \centering
        \includegraphics[width=\textwidth]{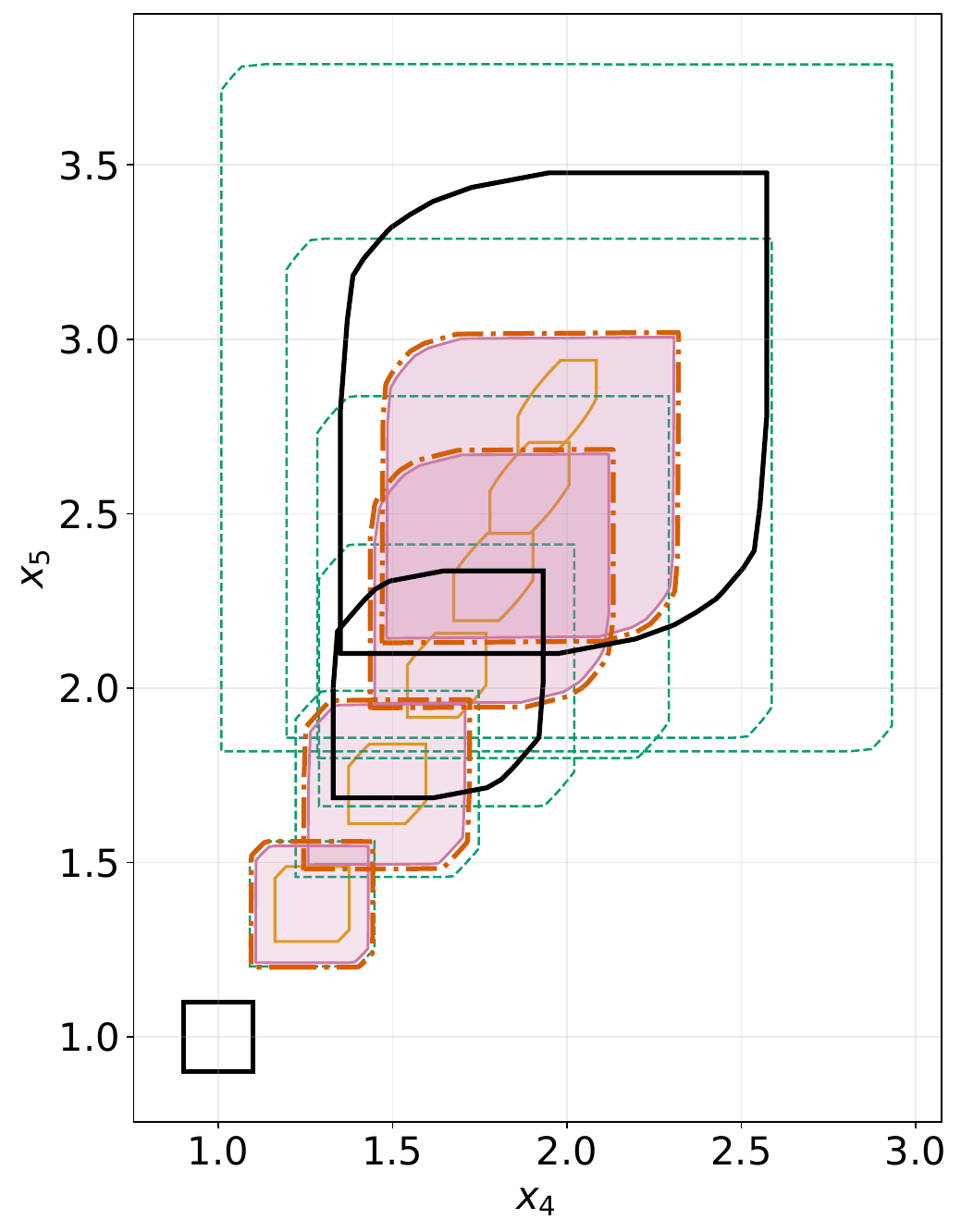}
        \caption{}
        \label{fig:ta_ira_x4x5}
    \end{subfigure}
    \caption{TA-IRA interpolation with Transformer predictions and conformal calibration for the same projections as Fig.~\ref{fig:mb_comparison}.}
    \label{fig:ta-ira-conv}
\end{figure*}

We evaluate the proposed IRA and TA-IRA framework on a five-dimensional linear system discretized from the continuous-time model used in~\cite{alanwar2023data}, with 
$A_c = \mathrm{blkdiag}\!\left(\left[\begin{matrix} -1 & -4 \\ 4 & -1 \end{matrix}\right], \left[\begin{matrix} -3 & 1 \\ -1 & -3 \end{matrix}\right], -2\right)$ and $B_c = 1_{5 \times 1}$.
The initial zonotope is $\mathcal{X}_0 = \zono{1_{5\times 1},\; 0.1 I_5}$, the input set is $\mathcal{U} = \zono{10,\; 0.25}$, and the noise intensity is $\sigma_w = 0.005$. The system is discretized with $\Delta_f = 0.05$\,s and $\Delta_c = 0.15$\,s, yielding $N_s = 3$ substeps per coarse interval. We collect $T = 150$ data samples from a single trajectory, and reduce zonotopes to order $4$ ($\kappa = 20$ generators) after each propagation step.

The prediction horizon spans $K = 2$ coarse steps, corresponding to $N_s K = 6$ fine steps over $T_\mathrm{total} = 0.30$\,s. Fig.~\ref{fig:mb_comparison} overlays the model-based chain and IRA results. The model-based sets at the coarse step size $\Delta_c = 0.15$\,s coincide with the fine propagation at every shared physical time, confirming step-size invariance with $\mathcal{U}^c = \mathcal{U}^f = \mathcal{U}$. The data-driven coarse and fine sets diverge visibly, as predicted by Proposition~\ref{prop:dd_depend}, which motivates the IRA construction.

The Transformer uses $d_\mathrm{model} = 256$, $8$ heads, $4$ encoder and decoder layers, FFN width $1024$, and is trained on $50{,}000$ samples for $1{,}000$ epochs with learning rate $3 \times 10^{-4}$. Split conformal calibration on a $15\%$ held-out set with $1 - \delta = 0.95$ yields $\hat{q} = 0.013$. Fig.~\ref{fig:ta-ira-conv} shows that the Transformer predictions $\hat{\mathcal{R}}^{\mathrm{TF}}$ closely track the fine DD chain, and the conformally inflated sets $\tilde{\mathcal{R}}^{\mathrm{TF}}$ visibly enclose the raw predictions.

\begin{table}[t]
\centering
\caption{Per-step wall-clock cost (ms) for $K{=}2$, $N_s{=}3$. Steps marked $\dagger$ run in parallel.}
\label{tab:timing}
\setlength{\tabcolsep}{4pt}
\begin{tabular}{l|cccccc|cc}
\hline
 & \multicolumn{6}{c|}{Fine-step index} & & \\
Method & 1 & 2 & 3 & 4 & 5 & 6 & Total & Speedup \\
\hline
DD & 1.20 & 213.8 & 414.3 & 409.0 & 408.3 & 408.0 & 1854.7 & $1.0\times$ \\
IRA & 309$^{\dagger}$ & 309$^{\dagger}$ & 1.2 & 309$^{\dagger}$ & 309$^{\dagger}$ & 206.9 & 826.4 & $2.2\times$ \\
TA-IRA & 1.9$^{\dagger}$ & 1.9$^{\dagger}$ & 1.2 & 1.9$^{\dagger}$ & 1.9$^{\dagger}$ & 206.9 & 211.8 & $8.8\times$ \\
\hline
\end{tabular}
\end{table}

Table~\ref{tab:timing} details the per-step wall-clock cost at the benchmark configuration $K = 2$, $N_s = 3$. The fine DD chain requires $1914$\,ms in total; IRA with parallel depth $K + N_s - 1 = 3$ reduces this to $843$\,ms, a $2.3\times$ speedup. TA-IRA further reduces each interpolation step to a fixed-cost forward pass of $3.4$\,ms, yielding a total of $221$\,ms and an $8.7\times$ speedup. The Transformer inference cost is independent of zonotope complexity because the tokenized representation fixes the sequence length at $\kappa + 1$ tokens.

\subsection{Runtime--Conservatism Trade-off}

Table~\ref{tab:tradeoff} sweeps $(K, N_s)$ and reports wall-clock time, speedup, three width ratios, and the Hausdorff distance between IRA and fine DD sets. Speedup grows with both $K$ and $N_s$ as predicted by the parallel depth model. The IRA/DD ratio stays below $1.0$ throughout, and the DD/MB and IRA/MB ratios, where MB denotes the model-based chain, exceed $1.0$ in all configurations, confirming that both data-driven methods remain valid over-approximations. These results indicate that IRA accelerates computation while offering the potential to reduce conservatism relative to the fine DD chain through the anchor-reset mechanism.

\begin{table}[t]
\centering
\caption{Runtime--conservatism trade-off of IRA for varying $K$ and $N_s$. Width ratios are mean IH width of the numerator divided by the denominator.}
\label{tab:tradeoff}
\setlength{\tabcolsep}{2.5pt}
\footnotesize
\begin{tabular}{ccrrrcccc}
\toprule
  & & \multicolumn{2}{c}{Time (ms)} & & \multicolumn{3}{c}{Width ratio} & \\
\cmidrule(lr){3-4} \cmidrule(lr){6-8}
$K$ & $N_s$ & DD & IRA & Speedup & $\frac{\text{IRA}}{\text{DD}}$ & $\frac{\text{DD}}{\text{MB}}$ & $\frac{\text{IRA}}{\text{MB}}$ & Hausdorff \\
\midrule
2 & \multirow{4}{*}{2} &  1149 &   529 & $2.2\times$ & 0.94 & 2.7 & 2.6 &  0.11 \\
3 &                    &  1914 &   742 & $2.6\times$ & 0.91 & 4.0 & 3.5 &  0.34 \\
4 &                    &  2679 &   955 & $2.8\times$ & 0.87 & 6.2 & 5.0 &  0.83 \\
5 &                    &  3444 &  1169 & $3.0\times$ & 0.84 & 9.2 & 7.0 &  1.73 \\
\midrule
2 & \multirow{4}{*}{3} &  1914 &   843 & $2.3\times$ & 0.84 & 4.0 & 3.1 &  0.58 \\
3 &                    &  3062 &  1056 & $2.9\times$ & 0.75 & 7.5 & 4.8 &  1.96 \\
4 &                    &  4209 &  1270 & $3.3\times$ & 0.68 & 13.9 & 7.2 &  5.46 \\
5 &                    &  5357 &  1483 & $3.6\times$ & 0.62 & 26.3 & 11.2 & 14.83 \\
\midrule
2 & \multirow{4}{*}{4} &  2679 &  1157 & $2.3\times$ & 0.75 & 6.2 & 3.8 &  1.71 \\
3 &                    &  4209 &  1371 & $3.1\times$ & 0.62 & 13.9 & 5.9 &  6.63 \\
4 &                    &  5740 &  1584 & $3.6\times$ & 0.52 & 32.6 & 9.4 & 23.90 \\
5 &                    &  7270 &  1797 & $4.0\times$ & 0.45 & 79.5 & 15.0 & 81.63 \\
\bottomrule
\end{tabular}
\end{table}

Table~\ref{tab:ira_speedup} extends the analysis to the full $(K, N_s) \in \{2,\ldots,6\}^2$ parameter space. The speedup grows monotonically with both $K$ and $N_s$, consistent with the theoretical prediction $S = KN_s / (K + N_s - 1)$. IRA reaches a maximum of $5.1\times$ at $K = N_s = 6$, while TA-IRA achieves $18.2\times$ at $K = 2$, $N_s = 6$, representing a further $7.6\times$ improvement over IRA at the same operating point.

\begin{table}[t]
\centering
\caption{Speedup of IRA and TA-IRA over DD for varying $K$ and $N_s$.}
\label{tab:ira_speedup}
\setlength{\tabcolsep}{5pt}
\renewcommand{\arraystretch}{1.05}
\footnotesize
\begin{tabular}{cc|ccccc}
\toprule
$K$ & & $N_s{=}2$ & $3$ & $4$ & $5$ & $6$ \\
\midrule
\multirow{2}{*}{2} & IRA    & $2.2\times$ & $2.3\times$ & $2.3\times$ & $2.3\times$ & $2.4\times$ \\
                   & TA-IRA & $5.3\times$  & $8.7\times$  & $11.9\times$ & $15.1\times$ & $18.2\times$ \\
\midrule
\multirow{2}{*}{3} & IRA    & $2.6\times$ & $2.9\times$ & $3.1\times$ & $3.2\times$ & $3.3\times$ \\
                   & TA-IRA & $4.4\times$  & $7.0\times$  & $9.6\times$  & $12.1\times$ & $14.6\times$ \\
\midrule
\multirow{2}{*}{4} & IRA    & $2.8\times$ & $3.3\times$ & $3.6\times$ & $3.8\times$ & $4.0\times$ \\
                   & TA-IRA & $4.2\times$  & $6.5\times$  & $8.8\times$  & $11.1\times$ & $13.4\times$ \\
\midrule
\multirow{2}{*}{5} & IRA    & $2.9\times$ & $3.6\times$ & $4.0\times$ & $4.3\times$ & $4.6\times$ \\
                   & TA-IRA & $4.0\times$  & $6.2\times$  & $8.4\times$  & $10.6\times$ & $12.7\times$ \\
\midrule
\multirow{2}{*}{6} & IRA    & $3.0\times$ & $3.8\times$ & $4.4\times$ & $4.8\times$ & $5.1\times$ \\
                   & TA-IRA & $3.9\times$  & $6.1\times$  & $8.2\times$  & $10.3\times$ & $12.3\times$ \\
\bottomrule
\end{tabular}
\end{table}

\subsection{Ablation Study}

Table~\ref{tab:ablation} compares four variants at $K=2$, $N_s=3$, all measured as full-horizon wall-clock time including Phase~1. IRA-par reduces wall-clock cost relative to IRA-seq without changing set quality. TA-IRA without conformal calibration produces sets narrower than fine DD, with width ratio below $1.0$, because the Transformer learns a compressed approximation; adding the conformal correction $\hat{q}$ restores pointwise coverage while preserving the speed advantage. Phase~1 anchor computation remains the dominant sequential bottleneck.

\begin{table}[t]
\centering
\caption{Ablation study for $K=2$, $N_s=3$.}
\label{tab:ablation}
\setlength{\tabcolsep}{4pt}
\begin{tabular}{lrrrl}
\hline
Variant & Runtime (ms) & Speedup & Width ratio & Guarantee \\
\hline
IRA-seq              & 1044 & $1.83\times$ & 0.841 & det. \\
IRA-par              & 843 & $2.27\times$ & 0.841 & det. \\
TA-IRA (no $\hat{q}$) & 221 & $8.65\times$ & 0.773 & none \\
TA-IRA + conformal   & 221 & $8.65\times$ & 0.801 & stat. \\
\hline
Fine DD (ref)        & 1914 & $1.00\times$ & 1.000 & det. \\
\hline
\end{tabular}
\end{table}

\section{Conclusion} \label{sec:conclusion}

We presented Interpolated Reachability Analysis, a multi-resolution framework for data-driven reachability with unknown dynamics and bounded noise. IRA exploits the step-size sensitivity of data-driven reachable sets by computing coarse anchors and interpolating fine-resolution sets with formal over-approximation guarantees. A data-driven noise estimation procedure removes the need for continuous-time system knowledge. The Transformer-accelerated variant TA-IRA replaces fine-step propagation with a learned surrogate calibrated via conformal prediction for pointwise and path-wise coverage. The two modes offer complementary safety semantics: IRA provides deterministic set containment, while TA-IRA provides statistical coverage at significantly lower computational cost. Experiments on a five-dimensional system confirm that TA-IRA closely tracks fine-resolution reachable sets while reducing computation. Future work will extend the framework to nonlinear systems and higher-dimensional state spaces.

\bibliographystyle{IEEEtran}
\bibliography{ref}

\end{document}